\documentclass[onecolumn,letter,11pt,dvips,titlepage]{IEEEtran}

\usepackage{epsfig}
\usepackage{amsmath}
\usepackage{amsfonts}
\usepackage{bbm}
\usepackage{cite}
\usepackage{graphicx}
\usepackage{multirow}
\usepackage{longtable}
\usepackage{txfonts}
\usepackage{amsfonts} %!PN
\newtheorem{theorem}{Theorem}
\newtheorem{definition}{Definition}

\newtheorem{lemma}{Lemma}
\newtheorem{proposition}{Proposition}

\newtheorem{remark}{Remark}
\newtheorem{example}{Example}

\newcommand{\tabincell}[2]{\begin{tabular}{@{}#1@{}}#2\end{tabular}}
\def\BibTeX{{\rm B\kern-.05em{\sc i\kern-.025em b}\kern-.08em
    T\kern-.1667em\lower.7ex\hbox{E}\kern-.125emX}}

\begin{document}
 \title{ Period Distribution of Inversive Pseudorandom Number Generators Over Finite Fields
 }
 \author{ Bo Zhou$^{1,*}$, Qiankun Song$^{2}$\\
 \thanks{* Corresponding author.}
 \thanks{\quad E-mail addresses: zhoubocncq@163.com (B. Zhou), qiankunsong@163.com (Q. Song)}
 \vspace{2mm}
 \small{$^1$College of Information Science \& Engineering, Chongqing Jiaotong University, Chongqing 400074, P.R. China}\\
\small{$^2$Department of Mathematics, Chongqing Jiaotong University,
Chongqing 400074, P.R. China}}
 \date{}
 \maketitle

 \begin{abstract}
  In this paper, we focus on analyzing the period distribution of the inversive pseudorandom number generators (IPRNGs) over finite field $({\rm Z}_{N},+,\times)$, where $N>3$ is a prime. The sequences generated by the IPRNGs are transformed to $2$-dimensional linear feedback shift register (LFSR) sequences. By employing the generating function method and the finite field theory, the period distribution is obtained analytically. The analysis process also indicates how to choose the parameters and the initial values such that the IPRNGs fit specific periods. The analysis results show that there are many small periods if $N$ is not chosen properly. The experimental examples show the effectiveness of the theoretical analysis.
 \end{abstract}

 \emph{Keywords:} Inversive pseudorandom number generators (IPRNG); Linear feedback shift register (LFSR); Period distribution; Finite field.

\section{Introduction}
Pseudoramdom number generators (PRNGs) are deterministic algorithm that produces a long sequence of numbers that appear random and indistinguishable from a stream of
random numbers \cite{s17}, which are widely employed in science and engineering, such as Monte Carlo simulations, computer games and cryptography. In recent years, a variety of PRNGs based on nonlinear congruential method \cite{e8,k13}, chaotic maps \cite{j8,a1,a11} and linear feedback shift registers (LFSRs) \cite{r11,k12} are proposed. These PRNGs are implemented on finite state machines, which lead to the fact that sequence generated by them are ultimately periodic. In cryptographic applications, a long period is often required. Once the period is not long enough, the encryption algorithms may be vulnerable to attacks, e.g., in \cite{k12}, Kocarev \emph{et al.} proposed a public key encryption algorithms based on Chebyshev polynomials over the finite field, but in \cite{c4,c5}, Chen \emph{et al.} showed that if the period of the sequence generated by the Chebyshev polynomials is not sufficiently long, the public key encryption algorithm is easy to be decrypted. Therefore, it is worth to making clear that what are the possible periods of a PRNG and how to choose suitable control parameters and initial values such that the PRNG fits specific period, these knowledge helps in algorithm design and its related applications.

In \cite{c4,c5}, Chen \emph{et al.} analyzed the period distribution of the sequence generated by the Chebyshev polynomials over finite fields and integer rings, respectively, by employing the generating function method. In \cite{c3}, Chen \emph{et al.} analyzed the period distribution of the generalized discrete Arnold cat map over Galois rings by employing the generating function method and the Hensel lifting method. In \cite{c6}, Chen \emph{et al.} summarized their works on the period distribution of the sequence generated by the linear maps.

In \cite{c9}, Chou described all possible period lengths of IPRNG (1) and showed that these period lengths are related to the periods of some polynomials. However, the author did not give the full information on period distribution, this leads to the limitation of the applications of IPRNGs. In \cite{s7}, Sol\'{e} \emph{et al.}  proposed an open problem of arithmetic interest to study the period of the IPRNGs and to give conditions bearing on $a,b$ to achieve maximal period. Although their considered state space is a Galois ring, it is also significant to study this problem in finite field. Recent results on the distribution property in parts of the period of this generator over finite fields can be found in \cite{g16,n17} and it would be interesting to generalize these results to arbitrary parts of the period. If the the full information on the period distribution is known, we could do such a work.

Motivated by the above discussions, we focus on analyzing the period distribution of the IPRNGs over the finite field $({\rm Z}_{N},+,\times)$, where $N>3$ is a prime. The analysis process is that, first, to make exact statistics on the periods of model (1), then count the number of IPRNGs for each specific period when $a$, $b$ and $x_{0}$ traverse all elements in ${\rm Z}_{N}$. The sequences generated by model (1) are transformed to $2$-dimensional LFSR sequences which is the foundation of the stream ciphers \cite{g21}. Then, the detailed period distribution of IPRNGs is obtained by employing the generating function method and the finite field theory. The analysis process also indicates how to choose the parameters and the initial values such that the IPRNGs fit specific periods.

This paper is organized as follows. To make this paper self-contained, Section II presents some preliminaries that help to understand our analysis. In Section III, detailed analysis of the period distribution of the sequences generated by IPRNGs with $ab=0$ in ${\rm Z}_{N}$ and $x_{0}\in{\rm Z}_{N}$. Then Section IV presents the detailed analysis of the period distribution of the sequences generated by IPRNGs with $a \in {\rm Z}^{\times }_{N}$, $b \in {\rm Z}^{\times }_{N}$ and $x_{0}\in{\rm Z}_{N}$. Finally, conclusion and some suggestions for future work are made in Section V.
\section{Preliminaries}
In this section, we introduce relevant notation and definition to facilitate the presentation of main results in the ensuing sections. For the knowledge of finite fields, please refer to \cite{l2}.

\subsection{Recurring relation over the finite field}
Let ${\rm Z}_{N}$ be the residue ring of integers modulo $N$. When $N$ is prime, $({\rm Z}_{N},+,\times)$ forms a finite field to which the modular operation is required in addition and multiplication.
\begin{definition} \cite{l2}.
A sequence $a_{0},a_{1},\ldots$ satisfying the relation over $({\rm Z}_{N},+,\times)$:
\begin{eqnarray}
a_{n+k}=c_{1}a_{n+k-1}+c_{2}a_{n+k-2}+\ldots+c_{k}a_{n} \quad {\rm mod}N,
\end{eqnarray}
where $c_{i}\in {\rm Z}_{N}$ for all $i=1,2,\ldots$, is called a linear recurring sequence in ${\rm Z}_{N}$.
\end{definition}

The generation of the linear recurring sequences can be implemented on a linear feedback shift register which is a special kind of electronic switching circuit handling information in the form of elements in ${\rm Z}_{N}$.

\begin{definition} \cite{l2}.
$f(t)=t^{k}-c_{1}t^{k-1}-\cdots-c_{k}$ is called the characteristic polynomial of recurring relation (1). Also, the sequence $a_{0},a_{1},\ldots$ is called the sequence generated by $f(t)$ in ${\rm Z}_{N}$.
\end{definition}

The characteristic polynomial $f(t)$ plays an important role in analyzing the period of the sequence generated by recurring relation (1). It follows from \cite{c4} that if all roots of $f(t)$ are with multiplicity $1$, then the period $T$ of $a_{0},a_{1},\ldots$ equals to ${\rm per}(f)$. ${\rm per}(f)$ is the smallest integer such that $f(t)\mid t^{{\rm per}(f)}-1$, which is called the period of $f(t)$. Then, we have the following proposition on ${\rm per}(f)$.
\begin{proposition}
If $f(t)$ can be factorized as $f(t)=(t-\alpha_{1})(t-\alpha_{2})\ldots(t-\alpha_{m})$, where $\alpha_{i}\neq\alpha_{j}$ for all $1\leq i,j\leq m$ and $i\neq j$, then ${\rm per}(f)={\rm lcm}({\rm ord}(\alpha_{1}),{\rm ord}(\alpha_{2}),\ldots,{\rm ord}(\alpha_{m}))$, where ${\rm lcm}({\rm ord}(\alpha_{1}),{\rm ord}(\alpha_{2}),\ldots,{\rm ord}(\alpha_{m}))$ is the least common multiple of ${\rm ord}(\alpha_{1}),{\rm ord}(\alpha_{2}),\ldots,{\rm ord}(\alpha_{m})$.
\end{proposition}
\begin{proof}
Let $L={\rm lcm}({\rm ord}(\alpha_{1}),{\rm ord}(\alpha_{2}),\ldots,{\rm ord}(\alpha_{m}))$. Since $\alpha^{L}_{i}-1=0$ for all $i=1,2,\ldots,m$, it is valid that
$$
t-\alpha_{i}\mid t^{L}-1
$$
for all $i=1,2,\ldots,m$. Since $\alpha_{i}\neq\alpha_{j}$ for all $1\leq i,j\leq m$ and $i\neq j$, it is valid that $t-\alpha_{i}$ and $t-\alpha_{j}$ are coprime for all $i,j$. Thus, $(t-\alpha_{1})(t-\alpha_{2})\ldots(t-\alpha_{m})\mid t^{L}-1$, which means that $f(t)\mid t^{L}-1$. By the property of the order, we have ${\rm per}(f)=L$. The proof is completed.
\end{proof}

In \cite{c4,c5}, Proposition 1 is employed to analyze the period distributions of two linear maps: the Chebyshev map and the generalized discrete cat map, whose characteristic polynomials can be expressed as $f(t)=t^{2}+a t+1 \in {\rm Z}_{N}[t]$, where $N$ is an integer. If $\alpha$ and $\beta$ are roots of $f(t)$, then it must hold that $\alpha\beta=1$. Thus, ${\rm ord}(\alpha)={\rm ord}(\beta)$. By Proposition 1, we have ${\rm per}(f)={\rm ord}(\alpha)$, so $T={\rm ord}(\alpha)$. However, if the characteristic polynomial is $f(t)=t^{2}+a t+b \in {\rm Z}_{N}[t]$, whose roots are $\alpha$ and $\beta$, where $b\neq 1$, we can not conclude that $ {\rm ord}(\alpha)={\rm ord}(\beta)$. In order to analyze the period $T$, we should analyze $ {\rm ord}(\alpha)$ and ${\rm ord}(\beta)$, respectively. If $N$ is not chosen properly, i.e., both $N-1$ and $N+1$ has many divisors, the analysis process is rather complicated. This obstacle prompts us to adopt another approach which will be presented in Section IV.
\subsection{IPRNGs over the finite field}
In this paper, we consider the following IPRNG proposed in \cite{e8} over $({\rm Z}_{N},+,\times)$:
\begin{eqnarray}
x_{n+1}=\left\{\begin{array}{cccc}
a x_{n}^{-1}+b  &x_{n}\in{\rm Z}^{\times}_{N}\\
b&x_{n}=0
\end{array}\right.,
\end{eqnarray}
for all $n\geq 1$, where $N>3$ is a prime, $a,b\in {\rm Z}_{N}$. The initial value associated with model (2) is given by $x_{0}\in{\rm Z}_{N}$.

Hereafter, we denote $S(x_{0};a,b)$ as the sequence generated by model (2) starts from $x_{0}$ for given $a$, $b$. Then, we have the following definition on the period of $S(x_{0};a,b)$.
\begin{definition}
For every initial value $x_{0}\in{\rm Z}_{N}$, the smallest integer $L(x_{0};a,b)$ such that $x_{n+L(x_{0};a,b)}=x_{n}$ for all $n\geq n_{0}\geq0$ is called the period of the IPRNGs correspond to $a$, $b$ and $x_{0}$, where $n_{0}$ is a nonnegative integer.
\end{definition}

\begin{remark}
It is noteworthy that the sequence generated by the IPRNGs may not be purely periodic, i.e. every period start from $x_{0}$, which is different from the case for the Chebyshev map and the generalized discrete Arnold cat map. Its period depends on not only the control parameters $a,b$ but also the initial value $x_{0}$, this will be illustrated in Section III and Section IV.
\end{remark}

Throughout this paper, ${\rm Z}_{N}$ denotes the residue ring of integers modulo $N$. ${\rm Z}^{\times}_{N}$ denotes the group of all units in ${\rm Z}_{N}$.
$({\rm Z}_{N},+,\times)$ denotes the finite field where addition and multiplication are all modular operations. For $\alpha\in{\rm Z}_{N}$, denote ${\rm ord}(\alpha)$ as the order of $\alpha$ in ${\rm Z}_{N}$. ${\rm GF}(N^{2})$ denotes a finite field with $N^{2}$ elements. $\varphi(n)$, i.e., Euler¡¯s totient function, denotes the number of positive integers which are both less than or equal to the positive integer and coprime with $n$.
\section{Period distribution of IPRNGs with $ab=0$ in ${\rm Z}_{N}$ and $x_{0}\in {\rm Z}_{N}$ }
When $ab=0$ in ${\rm Z}_{N}$ and $x_{0}\in{\rm Z}_{N}$, there are $2N^{2}-N$ IPRNGs. It would be better if we have an impression on what the period distribution with $ab=0$ in $ {\rm Z}_{N}$ and $x_{0}\in {\rm Z}_{N}$ looks like. Fig. 1 is a plot of the period distribution of IPRNGs (2) with $ab=0$ in ${\rm Z}_{31}$ and $x_{0}\in {\rm Z}_{31}$. It can be seen from Fig. 1 that the periods distribute very sparsely, some exist and some do not.
\begin{figure}[!ht]
%\noindent
%\begin{minipage} [b] {0.50\linewidth}
\centering\epsfig{figure=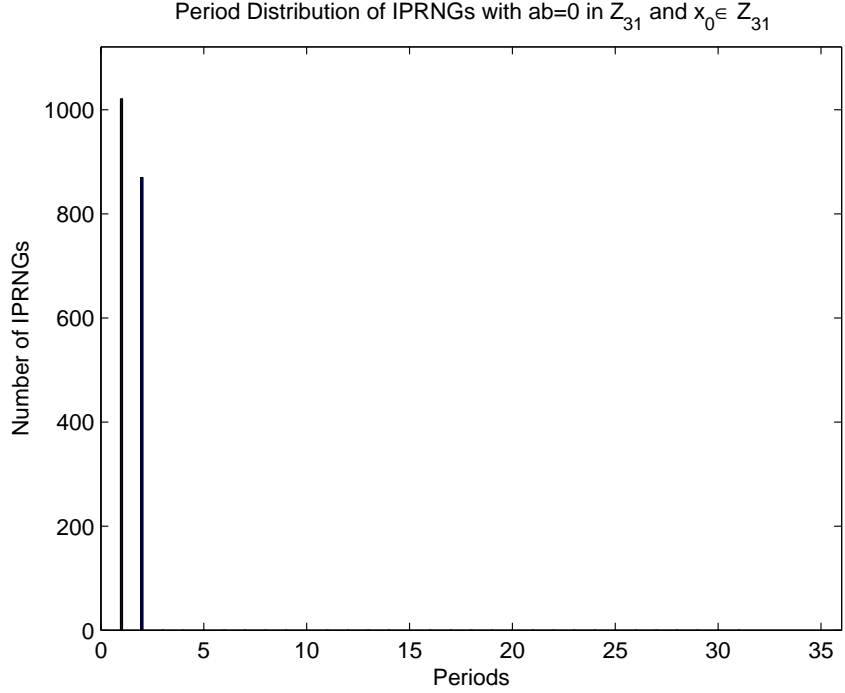, width=0.7\linewidth}
\caption{{\footnotesize Period distribution of IPRNGs with $ab=0$ in ${\rm Z}_{31}$ and $x_{0}\in {\rm Z}_{31}$.}} \label{Figure}
%\end{minipage}
\hfill
\end{figure}

In \cite{c9}, Chou has considered the periods of IPRNGs for $ab=0$ in ${\rm Z}_{N}$ and $x_{0}\in {\rm Z}_{N}$. The results are listed as follows
\begin{proposition}
Suppose $a=0$, then $x_{n}=b$ for all $n\geq 1$ and $L(x_{0};0,b)=1$.
\end{proposition}
\begin{proposition}
Suppose $a\neq 0$ and $b=0$.

(P1) If $x_{0}=0$, then $x_{n}=0$ for all $n\geq 1$ and $L(0;a,b)=1$.

(P2) If $a=x_{0}^{2}$ and $x_{0}\neq0$, then $x_{n}=x_{0}$ for all $n\geq1$ and $L(x_{0};x^{2}_{0},b)=1$.

(P3) If $a\neq x_{0}^{2}$ and $x_{0}\neq0$, then $x_{n+2}=x_{n}$ for all $n\geq1$ and $L(x_{0};a,b)=2$.
\end{proposition}

Now, all the possible periods for this case are revealed. In the following, we will count the number of IPRNGs for each specific period and present the period distribution.
\begin{theorem}
For IPRNG (2) with $ab=0$ in ${\rm Z}_{N}$ and $x_{0}\in {\rm Z}_{N}$ , the possible periods and the number of each special period are given in Table I.
\begin{table}[!t]
\renewcommand{\arraystretch}{2}
\caption{Period distribution of IPRNGs with $ab=0$ in ${\rm Z}_{N}$ and $x_{0}\in {\rm Z}_{N}$.}
\label{table_example}
\centering
\begin{tabular}{|c|c|}
\hline
\bfseries Periods & \bfseries Number of IPRNGs\\
\hline
 \tabincell{c}{$1$} & $N^{2}+2N-2$\\
\hline
\tabincell{c}{$2$}& \tabincell{c}{$(N-2)(N-1)$}\\
\hline
\end{tabular}
\end{table}
\end{theorem}
\begin{proof}
For $L(x_{0};a,b)=1$, there are three cases:

(i) $a=0$. Here, the choice of $a$ is unique and there are $N$ choices of $b$ and $N$ choices of $x_{0}$. Thus, there are $N^{2}$ IPRNGs.

(ii) $a\neq 0$, $b=0$ and $x_{0}=0$. Here, there are $N-1$ choices of $a$ and the choices of $b$ and $x_{0}$ are unique. Thus, there are $N-1$ IPRNGs.

(iii) $a\neq 0$, $b=0$ and $a=x^{2}_{0}$. Here, there is a unique choice of $b$. Since $a\neq 0$ and $a=x^{2}_{0}$, it is valid that $x_{0}\neq0$. Thus, there are $N-1$ choices of $x_{0}$. Once $x_{0}$ is chosen, $a$ is uniquely determined. Thus, there are $N-1$ IPRNGs.

Combining (i), (ii) and (iii), we have there are $N^{2}+2N-2$ IPRNGs for $L(x_{0};a,b)=1$.

For $L(x_{0};a,b)=2$, since $x_{0}\neq0$, there are $N-1$ choices of $x_{0}$. Once $x_{0}$ is chosen, combining $a\neq0$, there are $N-2$ choices of $a$ and a unique choice of $b$. Thus, there are $(N-2)(N-1)$ IPRNGs. The proof is completed.
\end{proof}
\begin{example}
The following example is given to compare experimental and the theoretical results. A computer program has been written to exhaust all possible IPRNGs with $ab=0$ in ${\rm Z}_{31}$ and $x_{0}\in {\rm Z}_{31}$ to find the period by brute force, the results are shown in Fig. 1.

Table II lists the complete result we have obtained. It provides the period distribution of the IPRNGs. As it is shown in Fig. 1 and Table II, the theoretical and experimental results fit well. The maximal period is $2$ while the minimal period is $1$. The analysis process also indicates how to choose the parameters and the initial values such that the IPRNGs fit specific periods.
\begin{table}[!t]
\renewcommand{\arraystretch}{2}
\caption{Period distribution of IPRNGs with $ab=0$ in ${\rm Z}_{31}$ and $x_{0}\in {\rm Z}_{31}$.}
\label{table_example}
\centering
\begin{tabular}{|c|c|}
\hline
\bfseries Periods & \bfseries Number of IPRNGs\\
\hline
 \tabincell{c}{$1$} & $1021$\\
\hline
\tabincell{c}{$2$}& \tabincell{c}{$870$}\\
\hline
\end{tabular}
\end{table}
\end{example}
\section{Period distribution of IPRNGs with $a\in{\rm Z}^{\times}_{N}$, $b\in{\rm Z}^{\times}_{N}$ and $x_{0}\in {\rm Z}_{N}$ }
In \cite{c9}, Chou described all possible periods of the model (2) with $a\in{\rm Z}^{\times}_{N}$, $b\in{\rm Z}^{\times}_{N}$ and $x_{0}\in {\rm Z}_{N}$ and showed that these periods were related to the periods of several polynomials, see Theorem 2 and Theorem 4 in \cite{c9}. However, the author did not provide a feasible way to evaluate these periods. In the following, we will characterize the full information on the period distribution of sequences generated by IPRNG (2) with $a,b$ traverse all elements in ${\rm Z}^{\times}_{N}$ and $x_{0}$ traverses all elements in ${\rm Z}_{N}$.

When $a$, $b$ traverse all elements in ${\rm Z}^{\times}_{N}$ and $x_{0}$ traverse all elements in ${\rm Z}_{N}$, there are $(N-1)^{2}N$ IPRNGs. It would be better if we have an impression on what the period distribution with $a\in{\rm Z}^{\times}_{N}$, $b\in {\rm Z}^{\times}_{N}$ and $x_{0}\in {\rm Z}_{N}$ looks like. Fig. 2 is a plot of the period distribution of IPRNGs (2) with $a\in{\rm Z}^{\times}_{31}$, $b\in{\rm Z}^{\times}_{31}$ and $x_{0}\in {\rm Z}_{31}$. It can be seen from Fig. 2 that the periods distribute very sparsely, some exist and some do not. In the following, the period distribution rules for $a\in{\rm Z}^{\times}_{N}$, $b\in{\rm Z}^{\times}_{N}$ and $x_{0}\in{\rm Z}_{N}$ will be worked out analytically.
\begin{figure}[!ht]
%\noindent
%\begin{minipage} [b] {0.50\linewidth}
\centering\epsfig{figure=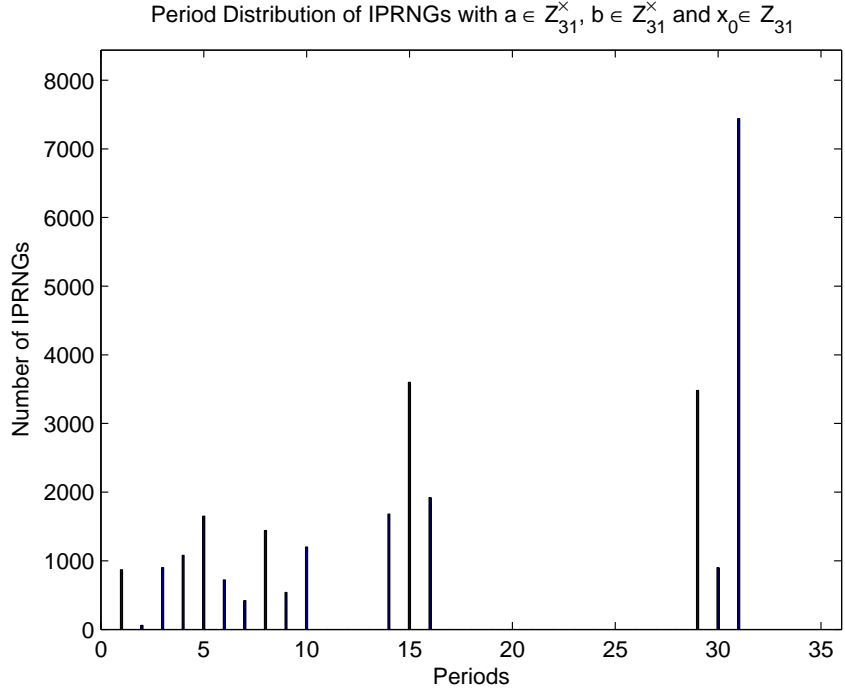, width=0.7\linewidth}
\caption{{\footnotesize Period distribution of IPRNGs with $a\in {\rm Z}^{\times}_{31}$, $b\in {\rm Z}^{\times}_{31}$ and $x_{0}\in {\rm Z}_{31}$.}} \label{Figure}
%\end{minipage}
\hfill
\end{figure}

In order to get the main results in the rest of this paper, we provide an important lemma in \cite{c9} which transforms the sequence generated by IPRNGs to 2-dimensional LFSR sequences.
\begin{lemma}\cite{c9}.
Let $a$, $b$ and $x_{0}$ are in ${\rm Z}_{N}$. Define the LFSR
\begin{eqnarray}
y_{n+2}=by_{n+1}+ay_{n},
\end{eqnarray}
for all $n\geq0$, where $y_{0}=1$, $y_{1}=x_{0}$. Then if $m\geq0$ is an integer such that $y_{n}\in {\rm Z}^{\times}_{p^{e}}$ for all $0\leq n\leq m$, then $x_{n}=y_{n+1}y^{-1}_{n}$ for all $0\leq n\leq m$. Moreover, $m$ is the smallest positive integer satisfying $x_{m}=$ if and only if $m+1$ is the smallest integer satisfying $y_{m+1}=0$.
\end{lemma}
Let $f(t)=t^{2}-bt-a$ be the characteristic polynomial of LFSR (3). If $f(t)$ has a root with multiplicity $2$, i.e., $f(t)=(t-\alpha)^{2}$, then $a=-\alpha^{2}$ and $b=2\alpha$. It follows from (3) that
\begin{eqnarray}
y_{n+2}= 2\alpha y_{n+1}-\alpha^{2}y_{n}.
\end{eqnarray}
By simple calculation, we can get the general term of (4)
\begin{eqnarray}
y_{n}=\alpha^{n}(n(\alpha^{-1}x_{0}-1)+1).
\end{eqnarray}

If $f(t)$ has two distinct roots with multiplicity $1$, i.e., $f(t)=(t-\alpha)(t-\beta)$ and $\alpha\neq\beta$, then $a=-\alpha\beta$ and $b=\alpha+\beta$. It follows from (3) that
\begin{eqnarray}
y_{n+2}= (\alpha+\beta) y_{n+1}-\alpha\beta y_{n}.
\end{eqnarray}
By simple calculation, we can get the general term of (6)
\begin{eqnarray}
y_{n}=(\alpha-\beta)^{-1}((x_{0}-\beta)\alpha^{n}+(\alpha-x_{0})\beta^{n}).
\end{eqnarray}

It can be observed from (5) and (7) that the general terms of (3) are different when $f(t)$ has a root with multiplicity $2$ and has two distinct roots with multiplicity $1$. Thus, we will discuss these two cases separately.
\subsection{$f(t)$ has a root with multiplicity $2$}
We suppose that $\alpha$ is a root of $f(t)$, i.e., $f(t)=(t-\alpha)^{2}$. In this case, it must holds that $\alpha\in{\rm Z}_{N}$. In fact, if $\alpha\notin {\rm Z}_{N}$, which means that $f(t)$ is irreducible in ${\rm Z}_{N}[t]$, then $f(t)$ must have two roots in ${\rm GF}(N^{2})$ and all roots of $f(t)$ are $\alpha$ and $\alpha^{N}$, where $\alpha$ and $\alpha^{N}$ are in ${\rm GF}(N^{2})$ but not in ${\rm Z}_{N}$. Since $f(t)$ has a root with multiplicity $2$, it must hold that $\alpha^{N}=\alpha$. Thus, $\alpha^{N-1}=1$, which means that ${\rm ord}(\alpha)\mid N-1$. Therefore, $\alpha\in {\rm Z}_{N}$, which is a contradiction.

It follows from (5) that if $x_{0}\neq\alpha$, then $y_{n}$ must contain $0$, which means that $S(x_{0};a,b)$ must contain some elements in $0$; Otherwise, $y_{n}$ dose not contain $0$, which means that $S(x_{0};a,b)$ does not contain $0$.

\begin{proposition}
Suppose $f(t)$ has a root with multiplicity $2$ in ${\rm Z}_{N}$. If $x_{0}\neq\alpha$, then $L(x_{0};a,b)=N-1$ and there are $(N-1)^{2}$ IPRNGs of period $N-1$.
\end{proposition}
\begin{proof}
Period analysis.

Since $x_{0}\neq\alpha$, it is valid that $y_{n}$ must contain $0$. Thus, $L(x_{0};a,b)=L(b;a,b)$. When $x_{0}=2\alpha$, it follows from (5) that $y_{n}=(n+1)\alpha^{n}$. Thus, $n=N-1$ is the smallest integer such that $y_{n}=0$. By lemma 1, we have $N-2$ is the smallest integer such that $x_{N-2}=0$. Thus, $x_{N-1}=b$, which means that $L(b;a,b)=N-1$.

Counting.

When $\alpha$ traverses all elements in ${\rm Z}^{\times}_{N}$, there are $N-1$ choices of $\alpha$. Since $f(t)=(t-\alpha)^{2}$, it is valid that $a$ and $b$ are uniquely determined by a chosen $\alpha$. Also, it follows from $x_{0}\neq \alpha$ that there are $N-1$ choices of $x_{0}$. Thus, there are $(N-1)^{2}$ IPRNGs of period $N-1$. The proof is completed.
\end{proof}
\begin{proposition}
Suppose $f(t)$ has a root with multiplicity $2$ in ${\rm Z}_{N}[t]$. If $x_{0}=\alpha$, then $L(x_{0};a,b)=1$ and there are $N-1$ IPRNGs of period $1$.
\end{proposition}
\begin{proof}
Period analysis.

Since $x_{0}=\alpha$, it is valid that $y_{n}$ does not contain $0$. It follows from (5) that $y_{n}=\alpha^{n}$.
By lemma 1, we can get that $x_{n}=\alpha$ for all $n=1,2,\ldots$. Thus, $L(x_{0};a,b)=1$.

Counting.

When $\alpha$ traverses all elements in ${\rm Z}^{\times}_{N}$, there are $N-1$ choices of $\alpha$. Since $f(t)=(t-\alpha)^{2}$, it is valid that $a$ and $b$ are uniquely determined by a chosen $\alpha$. Also, it follows from $x_{0}=\alpha$ that there is a unique choice of $x_{0}$. Thus, there are $N-1$ IPRNGs of period $1$. The proof is completed.
\end{proof}
\subsection{$f(t)$ has two distinct roots with multiplicity $1$}
It follows from (7) that $y_{n}=0$ if and only if
\begin{eqnarray}
(x_{0}-\alpha)(x_{0}-\beta)^{-1}=(\alpha\beta^{-1})^{n}.
\end{eqnarray}

For presentation convenience, we denote set
$\Omega=\{\alpha\beta^{-1},(\alpha\beta^{-1})^{2},\ldots,(\alpha\beta^{-1})^{{\rm ord}(\alpha\beta^{-1})-1}\}$.

If $(x_{0}-\alpha)(x_{0}-\beta)^{-1}\in\Omega$, there exists $1\leq n\leq p-1$ such that (8) holds, thus, $S(x_{0};a,b)$ must contains some elements in $0$; if $(x_{0}-\alpha)(x_{0}-\beta)^{-1}\notin\Omega$, there does not exist any $n$ such that (10) holds, thus, $S(x_{0};a,b)$ does not contain any element in $0$.

On the other hand, if either $x_{0}-\alpha=0$ or $x_{0}-\beta=0$, then $y_{n}\neq 0$ for all $n=1,2,\ldots$, which means that $S(x_{0};a,b)$ does not contain any element in $0$.

In the following, we will provide three lemmas which are necessary for our analysis.
\begin{lemma}
Suppose $a\in{\rm Z}_{N}^{\times}$, $b\in{\rm Z}_{N}^{\times}$. Then, if $\alpha,\beta$ are two distinct roots of $f(t)$, then ${\rm ord}(\alpha\beta^{-1})>2$.
\end{lemma}
\begin{proof}
Since $b\in{\rm Z}^{\times}_{N}$ and $b=\alpha+\beta$, it holds that $\alpha+\beta\neq0$. Combining $\alpha-\beta\neq0$, we have $\alpha\beta^{-1}-\alpha^{-1}\beta\neq0$, which means that $\alpha\beta^{-1}\neq \alpha^{-1}\beta$. If ${\rm ord}(\alpha\beta^{-1})=1$, then it must hold that $\alpha\beta^{-1}=1$ and $\alpha\beta^{-1}=1$, which contradicts to $\alpha\beta^{-1}\neq \alpha^{-1}\beta$. If ${\rm ord}(\alpha\beta^{-1})=2$, then it follows from $\varphi(2)=1$. Thus, $\alpha\beta^{-1}= \alpha^{-1}\beta$, which is a contradiction. The proof is completed.
\end{proof}
\begin{lemma}
Suppose $a\in{\rm Z}_{N}^{\times}$, $b\in{\rm Z}_{N}^{\times}$. If $\alpha,\beta$ are two distinct roots of $f(t)$, then $\alpha\beta^{-1}$ and $\alpha^{-1}\beta$ are two roots of $g(t)=t^{2}+(a^{-1}b^{2}+2)t+1$.
\end{lemma}
\begin{proof}
Since $\alpha,\beta$ are two distinct roots of $f(t)$, it is valid that $a=-\alpha\beta$ and $b=\alpha+\beta$. Then, it is easy to verify that $\alpha\beta^{-1}$ and $\alpha^{-1}\beta$ are roots of $g(t)$. The proof is completed.
\end{proof}
\begin{lemma}
Suppose $a\in{\rm Z}_{N}^{\times}$, $b\in{\rm Z}_{N}^{\times}$. If $\alpha,\beta$ are two distinct roots of $f(t)$, then $a^{-1}b^{2}$ is uniquely determined by $\alpha\beta^{-1}$.
\end{lemma}
\begin{proof}
Since $\alpha\beta^{-1}$ and $\alpha^{-1}\beta$ are roots of $g(t)$, it holds that $a^{-1}b^{2}+2=\alpha\beta^{-1}+\alpha^{-1}\beta$.

If $a^{-1}b^{2}$ is not uniquely determined by $\alpha\beta^{-1}$ or $\alpha^{-1}\beta$, then there exist $\alpha_{1}\beta^{-1}_{1}$ and $\alpha_{2}\beta^{-1}_{2}$ with $\alpha_{1}\beta^{-1}_{1}\neq \alpha_{2}\beta^{-1}_{2}$ and $\alpha_{1}\beta^{-1}_{1}\neq (\alpha_{2}\beta^{-1}_{2})^{-1}$, such that $\alpha_{1}\beta^{-1}_{1}+\alpha^{-1}_{1}\beta_{1}=\alpha_{2}\beta^{-1}_{2}+\alpha^{-1}_{2}\beta_{2}$. Let $\gamma_{1}=\alpha_{1}\beta^{-1}_{1}$ and $\gamma_{2}=\alpha_{2}\beta^{-1}_{2}$, then we have $\gamma_{1}\neq\gamma^{-1}_{2}$ and $\gamma_{1}\neq\gamma_{2}$. However, by simple calculation, we have $\gamma_{1}+\gamma^{-1}_{1}=\gamma_{2}+\gamma^{-1}_{2}$ if and only if $ (\gamma_{1}\gamma_{2}-1)(\gamma_{1}-\gamma_{2})=0$,
which means that either $\gamma_{1}\gamma_{2}=1$ or $\gamma_{1}=\gamma_{2}$. These are the contradictions. The proof is completed.
\end{proof}

When $f(t)$ has a root with multiplicity $2$, its roots are in ${\rm Z}_{N}$. However, when $f(t)$ has two distinct roots with multiplicity $1$, its roots may be in ${\rm GF}(N^{2})$ but not in ${\rm Z}_{N}$. Therefore, it is nature to consider the the following two cases separetely: 1) $\alpha$ and $\beta$ are in ${\rm Z}_{N}$; 2) $\alpha$ and $\beta$ are in ${\rm GF}(N^{2})$ but not in ${\rm Z}_{N}$.
\subsubsection{$\alpha$ and $\beta$ are in ${\rm Z}_{N}$}
\begin{proposition}
Suppose $f(t)$ has two distinct roots with multiplicity $1$ in ${\rm Z}_{N}$. If $(x_{0}-\alpha)(x_{0}-\beta)\neq 0$ and $(x_{0}-\alpha)(x_{0}-\beta)^{-1}\in\Omega$, then $L(x_{0};a,b)$ traverses the set $\{k-1:k>2,k\mid N-1\}$. For each $k$, there are $(k-1)(N-1)\frac{\varphi(k)}{2}$ IPRNGs of period $k-1$.
\end{proposition}
\begin{proof}
Period analysis.

If $(x_{0}-\alpha)(x_{0}-\beta)^{-1}\in\Omega$, then $S(x_{0};a,b)$ must contain $0$. Thus, $L(x_{0};a,b)=L(b;a,b)$. Then, we consider the case that $x_{0}=b$, which means that $x_{0}=\alpha+\beta$. By (7), we have $y_{n}=0$ if and only if $(\alpha\beta^{-1})^{n+1}=1$. Thus, $n={\rm ord}(\alpha\beta^{-1})-1$ is the smallest integer such that $y_{n}=0$. By Lemma 1, we have $x_{n-1}=0$, thus, $x_{n}=b$, which means that $L(x_{0};a,b)={\rm ord}(\alpha\beta^{-1})-1$.

Since $\alpha\beta^{-1}\in {\rm Z}^{\times}_{N}$, it holds that ${\rm ord}(\alpha\beta^{-1})\mid N-1$ and ${\rm ord}(\alpha\beta^{-1})>2$. Hence, $L(x_{0};a,b)$ traverses the set $\{k-1:k>2,k\mid p-1\}$.

Counting.

For $L(x_{0};a,b)=k-1$, there are $k-1$ $x_{0}$'s such that $(x_{0}-\alpha)(x_{0}-\beta)^{-1}\in\Omega$. Thus, there are $k-1$ choices of $x_{0}$.

Since $\alpha\beta^{-1}$ and $\alpha^{-1}\beta$ are roots of $g(t)$, it holds that $a^{-1}b^{2}+2=\alpha\beta^{-1}+\alpha^{-1}\beta$. Thus, $a=b^{2}(\alpha\beta^{-1}+\alpha^{-1}\beta-2)$. By Lemma 4, we have $a^{-1}b^{2}$ is uniquely determined by $\alpha\beta^{-1}$. Thus, when ${\rm ord}(\alpha\beta^{-1})=k$, there are $\frac{\varphi(k)}{2}$ different $\alpha\beta^{-1}+\alpha^{-1}\beta-2$ 's. Thus, there are $\frac{\varphi(k)}{2}$ choices of $\alpha\beta^{-1}+\alpha^{-1}\beta-2 $.

As a result of ${\rm ord}(\alpha\beta^{-1})>2$, we have $\alpha\beta^{-1}+\alpha^{-1}\beta-2$ is a unit. The number of choices of $b$ is $N-1$. Once $b$ and $\alpha\beta^{-1}+\alpha^{-1}\beta-2$ are chosen, $a$ is uniquely determined. Hence, for each $k$, there are $(k-1)(N-1)\frac{\varphi(k)}{2}$ IPRNGs of period $k-1$. The proof is completed.
\end{proof}
\begin{proposition}
Suppose $f(t)$ has two distinct roots with multiplicity $1$ in ${\rm Z}_{N}$. If $(x_{0}-\alpha)(x_{0}-\beta)\neq 0$ and $(x_{0}-\alpha)(x_{0}-\beta)^{-1}\notin\Omega$, then $L(x_{0};a,b)$ traverses the set $\{k:2<k<N-1,k\mid N-1\}$. For each $k$, there are $(N-(k-1))(N-1)\frac{\varphi(k)}{2}$ IPRNGs of period $k-1$.
\end{proposition}
\begin{proof}
Period analysis.

If $(x_{0}-\alpha)(x_{0}-\beta)^{-1}\notin\Omega$, then $S(x_{0};a,b)$ does not contain $0$. It follows from Lemma 1 and (7) that $x_{n}=x_{0}$ if and only if
\begin{eqnarray}
(x_{0}-\alpha)(x_{0}-\beta)\alpha^{n}=(x_{0}-\alpha)(x_{0}-\beta)\beta^{n}.
\end{eqnarray}
Since $(x_{0}-\alpha)(x_{0}-\beta)\neq 0$, (9) is equivalent to $(\alpha\beta^{-1})^{n}=1$. Thus, $L(x_{0};a,b)={\rm ord}(\alpha\beta^{-1})$.

By lemma 2, we have  ${\rm ord}(\alpha\beta^{-1})>2$. On the other hand, since $(x_{0}-\alpha)(x_{0}-\beta)^{-1}\notin\Omega$, it must hold that $\alpha\beta^{-1}$ is not a primitive element in ${\rm Z}_{N}$, which means that ${\rm ord}(\alpha\beta^{-1})\neq N-1$ Hence, $L(x_{0};a,b)$ traverses the set $\{k:2<k<N-1,k\mid N-1\}$.

Counting.

For $L(x_{0};a,b)=k$, there are $N-(k-1)$ $x_{0}$'s such that $(x_{0}-\alpha)(x_{0}-\beta)^{-1}\notin\Omega$. Thus, there are $N-(k-1)$ choices of $x_{0}$.

Since $\alpha\beta^{-1}$ and $\alpha^{-1}\beta$ are roots of $g(t)$, it holds that $a^{-1}b^{2}+2=\alpha\beta^{-1}+\alpha^{-1}\beta$. Thus, $a=b^{2}(\alpha\beta^{-1}+\alpha^{-1}\beta-2)$. By Lemma 4, we have $a^{-1}b^{2}$ is uniquely determined by $\alpha\beta^{-1}$. Thus, when ${\rm ord}(\alpha\beta^{-1})=k$, there are $\frac{\varphi(k)}{2}$ different $\alpha\beta^{-1}+\alpha^{-1}\beta-2$ 's. Thus, there are $\frac{\varphi(k)}{2}$ choices of $\alpha\beta^{-1}+\alpha^{-1}\beta-2 $.

As a result of ${\rm ord}(\alpha\beta^{-1})>2$, we have $\alpha\beta^{-1}+\alpha^{-1}\beta-2$ is a unit. The number of choices of $b$ is $N-1$. Once $b$ and $\alpha\beta^{-1}+\alpha^{-1}\beta-2$ are chosen, $a$ is uniquely determined. Hence, for each $k$, there are $(N-(k-1))(N-1)\frac{\varphi(k)}{2}$ IPRNGs of period $k$. The proof is completed.
\end{proof}
\begin{proposition}
Suppose $f(t)$ has two distinct roots with multiplicity $1$ in ${\rm Z}_{N}$. If $(x_{0}-\alpha)(x_{0}-\beta)= 0$, then $L(x_{0};a,b)=1$ and  there are $(N-3)(N-1)$ IPRNGs of period $k$.
\end{proposition}
\begin{proof}
Period analysis.

If $(x_{0}-\alpha)(x_{0}-\beta)=0$, then $y_{n}=x^{n}_{0}$. Thus, $x_{n}=x_{0}$ for all $n=1,2,\ldots$, which means that $L(x_{0};a,b)=1$.

Counting.

For $L(x_{0};a,b)=1$, $\alpha,\beta$ traverses all suitable elements in ${\rm Z}^{\times}_{N}$, i.e. both $\alpha-\beta$ and $\alpha+\beta$ are units, there are $\frac{(N-3)(N-1)}{2}$ pairs of $\alpha,\beta$. Once $\alpha,\beta$ are chosen, there are $2$ choices of $x_{0}$. Thus, there are $(N-3)(N-1)$ IPRNGs of period $1$. The proof is completed.
\end{proof}
\subsubsection{$\alpha$ and $\beta$ are in ${\rm GF}(N^{2})$ but not in ${\rm Z}_{N}$}
In this case, it must hold that $(x_{0}-\alpha)(x_{0}-\beta)\neq0$. Then, we have the following results on the period distribution of IPRNGs for this case.
\begin{proposition}
Suppose $f(t)$ has two distinct roots with multiplicity $1$ in ${\rm GF}(N^{2})$ but not in ${\rm Z}_{N}$. If $(x_{0}-\alpha)(x_{0}-\beta)^{-1}\in\Omega$, then $L(x_{0};a,b)$ traverses the set $\{k-1:k>2,k\mid N+1\}$. For each $k$, there are $(k-1)(N-1)\frac{\varphi(k)}{2}$ IPRNGs of period $k-1$.
\end{proposition}
\begin{proof}
Period analysis.

If $(x_{0}-\alpha)(x_{0}-\beta)^{-1}\in\Omega$, then $S(x_{0};a,b)$ must contain $0$. Thus, $L(x_{0};a,b)=L(b;a,b)$. Then, we consider the case that $x_{0}=b$, which means that $x_{0}=\alpha+\beta$. By (7), we have $y_{n}=0$ if and only if $(\alpha\beta^{-1})^{n+1}=1$. Thus, $n={\rm ord}(\alpha\beta^{-1})-1$ is the smallest integer such that $y_{n}=0$. By Lemma 1, we have $x_{n-1}=0$, thus, $x_{n}=b$, which means that $L(x_{0};a,b)={\rm ord}(\alpha\beta^{-1})-1$.

By lemma 2, we have  ${\rm ord}(\alpha\beta^{-1})>2$. Since $\alpha\beta^{-1}\in {\rm GF}(N^{2})$, it holds that ${\rm ord}(\alpha\beta^{-1})\mid N^{2}-1$. Notice that $\alpha$ and $\beta$ are not in ${\rm Z}_{N}$ and $\alpha\neq\beta$, it is valid that $\alpha\beta^{-1}\notin{\rm Z}_{N} $. Since ${\rm Z}_{N}\subseteq {\rm GF}(N^{2})$, it is valid that all units in ${\rm Z}_{N}$ are contained in ${\rm GF}(N^{2})$, which means that ${\rm ord}(\alpha\beta^{-1})\nmid N-1$. Thus, ${\rm ord}(\alpha\beta^{-1})\mid N+1$. Hence, $L(x_{0};a,b)$ traverses the set $\{k-1:k>2,k\mid N+1\}$.

Counting.

For $L(x_{0};a,b)=k-1$, there are $k-1$ $x_{0}$'s such that $(x_{0}-\alpha)(x_{0}-\beta)^{-1}\in\Omega$. Thus, there are $k-1$ choices of $x_{0}$.

Since $\alpha\beta^{-1}$ and $\alpha^{-1}\beta$ are roots of $g(t)$, it holds that $a^{-1}b^{2}+2=\alpha\beta^{-1}+\alpha^{-1}\beta$. Thus, $a=b^{2}(\alpha\beta^{-1}+\alpha^{-1}\beta-2)$. By Lemma 4, we have $a^{-1}b^{2}$ is uniquely determined by $\alpha\beta^{-1}$. Thus, when ${\rm ord}(\alpha\beta^{-1})=k$, there are $\frac{\varphi(k)}{2}$ different $\alpha\beta^{-1}+\alpha^{-1}\beta-2$ 's. Hence, there are $\frac{\varphi(k)}{2}$ choices of $\alpha\beta^{-1}+\alpha^{-1}\beta-2 $.

As a result of ${\rm ord}(\alpha\beta^{-1})>2$, we have $\alpha\beta^{-1}+\alpha^{-1}\beta-2$ is a unit. The number of choices of $b$ is $N-1$. Once $b$ and $\alpha\beta^{-1}+\alpha^{-1}\beta-2$ are chosen, $a$ is uniquely determined. Hence, for each $k$, there are $(k-1)(N-1)\frac{\varphi(k)}{2}$ IPRNGs of period $k-1$. The proof is completed.
\end{proof}
\begin{proposition}
Suppose $f(t)$ has two distinct roots with multiplicity $1$ in ${\rm GF}(N^{2})$ but not in ${\rm Z}_{N}$. If $(x_{0}-\alpha)(x_{0}-\beta)^{-1}\notin\Omega$, then $L(x_{0};a,b)$ traverses the set $\{k:2<k<N+1,k\mid N+1\}$. For each $k$, there are $(N-(k-1))(N-1)\frac{\varphi(k)}{2}$ IPRNGs of period $k$.
\end{proposition}
\begin{proof}
Period analysis.

If $(x_{0}-\alpha)(x_{0}-\beta)^{-1}\notin\Omega$, then $S(x_{0};a,b)$ does not contain $0$. It follows from Lemma 1 and (7) that $x_{n}=x_{0}$ if and only if
\begin{eqnarray}
(x_{0}-\alpha)(x_{0}-\beta)\alpha^{n}=(x_{0}-\alpha)(x_{0}-\beta)\beta^{n}.
\end{eqnarray}
Since $(x_{0}-\alpha)(x_{0}-\beta)\neq 0$, (10) is equivalent to $(\alpha\beta^{-1})^{n}=1$. Thus, $L(x_{0};a,b)={\rm ord}(\alpha\beta^{-1})$.

By lemma 2, we have  ${\rm ord}(\alpha\beta^{-1})>2$. Since $\alpha\beta^{-1}\in {\rm GF}(N^{2})$, it holds that ${\rm ord}(\alpha\beta^{-1})\mid N^{2}-1$. Notice that $\alpha$ and $\beta$ are not in ${\rm Z}_{N}$ and $\alpha\neq\beta$, it is valid that $\alpha\beta^{-1}\notin{\rm Z}_{N} $. Since ${\rm Z}_{N}\subseteq {\rm GF}(N^{2})$, it is valid that all units in ${\rm Z}_{N}$ are contained in ${\rm GF}(N^{2})$, which means that ${\rm ord}(\alpha\beta^{-1})\nmid N-1$. Thus, ${\rm ord}(\alpha\beta^{-1})\mid N+1$.

On the other hand, since $(x_{0}-\alpha)(x_{0}-\beta)^{-1}\notin\Omega$, it must hold that $\alpha\beta^{-1}$ is not a primitive element in ${\rm GF}(N^{2})$, which means that ${\rm ord}(\alpha\beta^{-1})\neq N+1$ Hence, $L(x_{0};a,b)$ traverses the set $\{k:2<k<N+1,k\mid N+1\}$.

Counting.

For $L(x_{0};a,b)=k$, there are $N-(k-1)$ $x_{0}$'s such that $(x_{0}-\alpha)(x_{0}-\beta)^{-1}\notin\Omega$. Thus, there are $N-(k-1)$ choices of $x_{0}$.

Since $\alpha\beta^{-1}$ and $\alpha^{-1}\beta$ are roots of $g(t)$, it holds that $a^{-1}b^{2}+2=\alpha\beta^{-1}+\alpha^{-1}\beta$. Thus, $a=b^{2}(\alpha\beta^{-1}+\alpha^{-1}\beta-2)$. By Lemma 4, we have $a^{-1}b^{2}$ is uniquely determined by $\alpha\beta^{-1}$. Thus, when ${\rm ord}(\alpha\beta^{-1})=k$, there are $\frac{\varphi(k)}{2}$ different $\alpha\beta^{-1}+\alpha^{-1}\beta-2$ 's. Thus, there are $\frac{\varphi(k)}{2}$ choices of $\alpha\beta^{-1}+\alpha^{-1}\beta-2 $.

As a result of ${\rm ord}(\alpha\beta^{-1})>2$, we have $\alpha\beta^{-1}+\alpha^{-1}\beta-2$ is a unit. The number of choices of $b$ is $N-1$. Once $b$ and $\alpha\beta^{-1}+\alpha^{-1}\beta-2$ are chosen, $a$ is uniquely determined. Hence, for each $k$, there are $(N-(k-1))(N-1)\frac{\varphi(k)}{2}$ IPRNGs of period $k$. The proof is completed.
\end{proof}

Now, we summarize the results in the following theorem.
\begin{theorem}
For IPRNGs with $a\in{\rm Z}^{\times}_{N}$, $b\in{\rm Z}^{\times}_{N}$ and $x_{0}\in{\rm Z}_{N}$, the possible periods and the number of each special period are given in Table III.
\begin{table}[!t]
\renewcommand{\arraystretch}{2}
\caption{Period distribution of IPRNGs with $a\in{\rm Z}^{\times}_{N}$, $b\in{\rm Z}^{\times}_{N}$ and $x_{0}\in {\rm Z}_{N}$. }
\label{table_example}
\centering
\begin{tabular}{|c|c|}
\hline
\bfseries Periods & \bfseries Number of IPRNGs\\
\hline
 \tabincell{c}{$1$} & $(N-2)(N-1)$\\
\hline
\tabincell{c}{$N-1$}& \tabincell{c}{$(N-1)^{2}$}\\
\hline
\tabincell{c}{$\{k-1:k>2,k\mid N-1\}$} & \tabincell{c}{$(k-1)(N-1)\frac{\varphi(k)}{2}$}\\
\hline
\tabincell{c}{$\{k-1:k>2,k\mid N+1\}$}& \tabincell{c}{$(k-1)(N-1)\frac{\varphi(k)}{2}$}\\
\hline
\tabincell{c}{$\{k:2<k<N-1,k\mid N-1\}$} & \tabincell{c}{$(N-(k-1))(N-1)\frac{\varphi(k)}{2}$}\\
\hline
\tabincell{c}{$\{k:2<k<N+1,k\mid N+1\}$} & \tabincell{c}{$(N-(k-1))(N-1)\frac{\varphi(k)}{2}$}\\
\hline
\end{tabular}
\end{table}
\end{theorem}

\begin{remark}
It should be mentioned that $N>3$ is an important condition in Theorem 3, because of some periods require $k>2,k\mid N-1$, which implies that $N>3$.
\end{remark}

\begin{example}
The following example is given to compare experimental and the theoretical results. A computer program has been written to exhaust all possible IPRNGs with $a\in{\rm Z}^{\times}_{31}$ and $b\in{\rm Z}^{\times}_{31}$ and $x_{0}\in{\rm Z}_{31}$ to find the period by brute force, the results are shown in Fig. 2.

Table IV lists the complete result we have obtained. It provides the period distribution of the IPRNGs. As it is shown in Fig. 2 and Table IV, the theoretical and experimental results fit well. The maximal period is $31$ while the minimal period is $1$. The analysis process also indicates how to choose the parameters and the initial values such that the IPRNGs fit specific periods.
\begin{table}[!t]
\renewcommand{\arraystretch}{2}
\caption{Period distribution of IPRNGs with $a\in{\rm Z}^{\times}_{31}$, $b\in{\rm Z}^{\times}_{31}$ and $x_{0}\in {\rm Z}_{31}$}
\label{table_example}
\centering
\begin{tabular}{|c|c|c|c|c|c|c|c|c|c|}
\hline
Periods &1 &2 &3&4&5&6&7&8\\
\hline
Number of IPRNGs &870&60&900&1080&1650&720&420&1440\\
\hline
Periods &9&10&14 &15&16&29 &30&31\\
\hline
Number of IPRNGs &540&1200 &1680 &3600&1920&3480&900&7440 \\
\hline
\end{tabular}
\end{table}
\end{example}

\section{Conclusion}
The period distribution of the IPRNGs over $({\rm Z}_{N},+,\times)$ for prime $N>3$ has been analyzed. The period distribution of IPRNGs is obtained by the generating function method and the finite field theory. The analysis process also indicates how to choose the parameters and the initial values such that the IPRNGs fit specific periods. The analysis results show that the period distribution is poor if $N$ is not chosen properly and there are many small periods.

A feasible way to resolve the open problem proposed by Sol\'{e} \emph{et al.} in \cite{s7} is to analyze the period distribution of the sequence generated by IPRNGs over Galois rings. However, the period distribution of IPRNG sequences varies substantially as $N$ changes, when $N$ is a prime, $({\rm Z}_{N},+,\times)$ is a finite field; when $N$ is a power of prime, i.e., $N=p^{e}$, $({\rm Z}_{N},+,\times)$ is a Galois ring. The structure of $({\rm Z}_{p^{e}},+,\times)$ is more complicated than that of $({\rm Z}_{N},+,\times)$, because of $({\rm Z}_{p^{e}},+,\times)$ contains many zero divisors but $({\rm Z}_{N},+,\times)$ does not, this difference makes the fact that the analysis in Galois rings is more complicated than that in finite fields, which is challenging and deserves intensive study. Another important problem is to characterize the security properties of the IPRNGs. These topics are interesting and need further research.
\section*{Acknowledgements}
This work was partially supported by the National Natural Science Foundation of China under Grant 60974132, the Natural Science Foundation Project of CQ CSTC2011BA6026 and the Scientific \& Technological Research Projects of CQ KJ110424.

\ifCLASSOPTIONcaptionsoff
\newpage
\fi

 \end{document}